\documentclass[10pt]{article}
\usepackage{epsfig}
\usepackage{times}
\usepackage{natbib}
\usepackage{amscd}
\usepackage{amsmath}
\usepackage{amsthm}
\usepackage{amsfonts}
\usepackage{amssymb}
\usepackage{graphicx}
\usepackage{url}

\newcommand{\Above}[2]{\stackrel{\scriptstyle #1}{#2}}
\newcommand{\bm}[1]{{\bf #1}}
\newcommand{\0}{\bm{0}}
\newcommand{\A}{\bm{A}}
\newcommand{\B}{\bm{B}}

\newcommand{\C}{\bm{C}}

\newcommand{\D}{\bm{D}}
\newcommand{\I}{\bm{I}}
\newcommand{\Ic}{{\cal I}}

\newcommand{\Lm}{\bm{L}}

\newcommand{\McSolve}{\boldsymbol{{\cal M}}^{\star}}
\newcommand{\M}{\bm{M}}
\newcommand{\Mc}{\boldsymbol{\cal M}}

\newcommand{\Nc}{{\cal N}}
\newcommand{\Pc}{{\cal P}}

\newcommand{\Q}{\bm{Q}}
\newcommand{\RealPart}{\mbox{Re}}
\newcommand{\Sm}{\bm{S}}
\newcommand{\Smt}{\Sm^{\Join}}

\newcommand{\alphah}{\hat{\alpha}}
\newcommand{\bms}[1]{\mbox{\boldmath\(#1\)\unboldmath}}
\newcommand{\benu}{\begin{enumerate}}

\newcommand{\diag}[1]{\mbox{\bf diag}\matrx{#1}}
\newcommand {\dsty}{\displaystyle}
\newcommand{\da}[1]{\frac{d #1}{d \alpha}}
\newcommand{\db}[1]{\frac{d #1}{d \beta}}
\newcommand{\eban}{\begin{eqnarray*}}
\newcommand{\eba}{\begin{eqnarray}}
\newcommand{\eb}{\begin{equation}}
\newcommand{\ee}{\end{equation}}
\newcommand{\eean}{\end{eqnarray*}}
\newcommand{\eea}{\end{eqnarray}}
\newcommand{\eenu}{\end{enumerate}}
\newcommand{\ev}{\bm{e}}
\newcommand{\epsi}{\bms{\epsilon}}
\newcommand{\la}{\!\leftarrow\!}
\newcommand{\matrx}[1]{{\left[ \stackrel{}{#1}\right]}}
\newcommand{\ov}{\overline}
\newcommand{\pv}{\bm{p}}

\newcommand{\tp}{{\top}}	

\newcommand{\vv}{\bm{v}}
\newcommand{\wbh}{\widehat{\wb} } 
\newcommand{\wb}{\ov{w}} 
\newcommand{\wh}{\hat{w}} 
\newcommand{\x}{\bm{x}}
\newcommand{\y}{\bm{y}}
\newcommand{\zh}{\hat{z}}
\newcommand{\zvh}{\hat{\z}}
\newcommand{\z}{\bm{z}}
\newtheorem{Corollary}{Corollary}
\newtheorem{Lemma}{Lemma}
\newtheorem{Theorem}{Theorem}

\begin{document}

\title{The Evolutionary Reduction Principle for \\Linear Variation in Genetic Transmission}
\author{Lee Altenberg
\footnote{Dedicated to my doctoral advisor Marc Feldman on his 65th birthday, and to the memory of Marc's doctoral advisor, Sam Karlin, who both laid the foundations necessary for these results; and to my mother Elizabeth Lee and to the memory of my father, Roger Altenberg, who both laid the foundation necessary for me.}
\\{University of Hawai`i at Manoa}
\thanks{To whom correspondence should be addressed. E-mail: altenber@hawaii.edu}
}


\maketitle

\begin{abstract}
The evolution of genetic systems has been analyzed through the use of modifier gene models, in which a neutral gene is posited to control the transmission of other genes under selection.  Analysis of modifier gene models has found the manifestations of an  ``evolutionary reduction principle'':  in a population near equilibrium, a new modifier allele that scales equally all transition probabilities between different genotypes under selection can invade if and only if it reduces the transition probabilities.   Analytical results on the reduction principle have always required some set of constraints for tractability: limitations to one or two selected loci, two alleles per locus, specific selection regimes or weak selection, specific genetic processes being modified, extreme or infinitesimal effects of the modifier allele, or tight linkage between modifier and selected loci.  Here, I prove the reduction principle in the absence of any of these constraints, confirming a twenty-year old conjecture.  The proof is obtained by a wider application of Karlin's Theorem 5.2 (1982) and its extension to ML-matrices, substochastic matrices, and reducible matrices.
\ \\ \ \\
\end{abstract}

\section{Introduction}

Darwinian evolution occurs through the interaction of two fundamental processes:  (1) natural selection, i.e. differential survival and reproduction;  and (2) genetic transformation, i.e. change of genetic {\em content} during reproduction, which provides the variation upon which selection can act.  The principal genetic transformations are sexual reproduction, recombination, and mutation, while a growing list of other transformations includes gene conversion, methylation, deletions, duplications, insertions, transpositions, and other chromosomal alterations.

These two processes, augmented by a third --- the randomness of sampling in finite populations --- provide the basis for our causal explanations of the characteristics of organisms.  In its simplest version:  Transformation processes create new genetic states among offspring, and differential survival and reproduction of alternate genetic states results in the prevalence of states with the highest levels of survival and reproduction.  In a more sophisticated version:  evolutionary dynamics are the result of the joint action of selection,  transformation, and random sampling processes that move populations to distribute over certain regions of genotype space.

The simpler version of the explanation runs into a quandary when trying to explain traits that in themselves do not cause different levels of survival or reproduction, in particular, the traits that make up the molecular, cellular, and organismal machinery of the genetic transformation processes themselves.  Variation in the transformation machinery can produce differing distributions of offspring genotypes, without necessarily affecting the parent's survival or their quantity of offspring.  To understand the fate of variation in the genetic transformation machinery requires the more sophisticated version of evolutionary causation, the evolutionary dynamics of the joint action of selection, transformation, and sampling processes.

The earliest mathematical treatments of evolutionary dynamics in the 1920s `Modern Synthesis' of Darwinism and Mendelian genetics  \citep{Fisher:1922,Haldane:1924I,Wright:1931} dealt straightaway with many complex issues;  however, the first analysis of genetic variation in the genetic transformation processes themselves waited another thirty years for Kimura's \citeyearpar{Kimura:1956} analysis of a model of recombination modification.  The model examines the fate of a chromosomal alteration that eliminates recombination between two loci with a stable polymorphism that exhibits linkage disequilibrium.  When the alteration occurs in the chromosome that have above average fitness, it increases in frequency.  This result affirmed Fisher's \citeyearpar[p. 130]{Fisher:1930} assertion that ``the presence of pairs of factors in the same chromosome, the selective advantage of each of which reverses that of the other, will always tend to diminish recombination, and therefore to increase the intensity of linkage in the chromosomes of that species.''  This result was the first instance of what was to later be called the ``reduction principle'' for the evolution of genetic transformations \citep{Feldman:1972,Feldman:Christiansen:and:Brooks:1980,Altenberg:1984,Liberman:and:Feldman:1986:GRP,Altenberg:and:Feldman:1987}.

 \citet{Nei:1967} introduced, and partially analyzed, a model for the evolution of recombination in which recombination rates between two loci are modified by a third, neutral locus.  \citet{Feldman:1972} gave a complete linear stability analysis of the model under the assumption of additive, multiplicative, and symmetric viability selection regimes.  He found that for populations near a polymorphic equilibrium with linkage disequilibrium under selection and recombination, genetic variation for recombination would survive if and only if it reduced the rate of recombination between the loci under selection.  

Other transformation processes were analyzed with modifier gene models, in particular mutation and migration rates, and the reduction principle was found to emerge again \citep{Karlin:and:McGregor:1972:PNAS,Feldman:and:Balkau:1973,Balkau:and:Feldman:1973,Karlin:and:McGregor:1974,Feldman:and:Krakauer:1976,Feldman:Christiansen:and:Brooks:1980}.  (Note that there is more recent acceptance of the idea that spatial location of organisms, and other environmental conditions, may function formally as heritable traits subject to transformation processes, e.g. \citet{Schauber:Goodwin:Jones:and:Ostfeld:2007,Odling-Smee:2007}.  However, this concept follows quite naturally from the concept of generalized transmission that appears in models of cultural transmission and modifier genes (\citealt{Cavalli-Sforza:and:Feldman:1973:MCI,Karlin:and:McGregor:1974}; \citealt[pp. 15-16,  p. 178]{Altenberg:1984}).)

These first reduction results were derived under narrow constraints on the selection regime, number of modifier alleles, or number of alleles under selection.  Subsequent studies of modifiers of these three processes --- recombination, mutation, and migration --- have extended the result to cases of modifier polymorphisms and arbitrary selection coefficients \citep{Liberman:and:Feldman:1986:MMR,Liberman:and:Feldman:1986:GRP,Feldman:and:Liberman:1986,Liberman:and:Feldman:1989}, but are still restricted to two alleles per selected locus, or two demes in the case of migration modification.

The repeated emergence of the reduction result in modifier models of different processes led to a study of mathematical underpinnings that might be common to all of them \citep{Altenberg:1984}.  The approach taken was to represent all possible transmission processes --- in which recombination and mutation comprise special cases --- using a general bi-parental transmission matrix, consisting of probabilities $T(i\la j, k)$ that parental haplotypes $j$ and $k$ produce a gamete haplotype $i$.  Haplotype here refers to a gamete's genotype, or a gamete's contribution to a diploid genotype.  (The addition of spatial subdivisions in the case of migration modification leads to a slightly changed representation \citep[pp. 178-199]{Altenberg:1984}.)

This approach allows the particulars of the processes to be abstracted out of the model, and reveals that what all the models have in common is the nature of the {\em variation} in transmission produced by variation at the modifier locus.  All of the models fit the form
\[
T_\alpha (i \la j, k) = \alpha \; P(i\la j, k) \mbox{\ for } j, k \neq i,
\] 
where $\alpha$ is the modified parameter that represents an overall rate of transformation of haplotypes, over all haplotypes.  This form is referred to as {\it linear variation} \citep{Altenberg:1984,Altenberg:and:Feldman:1987} because the modifier gene scales all transmission probabilities between different haplotypes equally.  

In the perturbation analysis of the evolutionary dynamics of the modifier locus (to be explained in detail in the next section), the stability matrix under linear variation has the form: 
\eb
\label{eq:StabilityMatrix}
\M(\alpha,r)\D =
\left\{(1-\alpha)[(1-r)\I + r\Q] + \alpha[(1-r)\Sm + r \Smt]\right\}\D,
\ee
where $\alpha$ is the transmission parameter produced by the new modifier allele, $\Q$, $\Sm$, and $\Smt$ are stochastic matrices, $\D$ is a positive diagonal matrix, and $r$ is the rate of recombination between the modifier and the nearest locus under selection.  

The rare modifier allele will increase when rare at a geometric rate if and only if the spectral radius $\rho(\M(\alpha,r)\D)$ exceeds 1.  This formulation is extremely general, accommodating modifier models with arbitrary numbers of modifier alleles, numbers of loci and alleles per locus under selection, selection regimes, and transmission processes.

Although the model is extremely general, it is tractable thanks to a theorem of \citet{Karlin:1982}:

\begin{Theorem}[Karlin Theorem 5.2, { \citeyearpar[pp. 194--196]{Karlin:1982}}]
Let $\M$ be an arbitrary non-negative irreducible stochastic matrix.  Consider the family of matrices
 \[
 \M{(\alpha)} = (1-\alpha) \I + \alpha \M
 \]
Then for any diagonal matrix $\D$ with positive terms on the diagonal, the spectral radius
\[
\rho(\alpha) = \rho( \M{(\alpha)} \D )
\]
is decreasing as $\alpha$ increases (strictly provided $\D \neq d \I$).
\end{Theorem}

We see that $\M(\alpha,r)\D$ fits the form in Karlin's theorem only if $r=0$, in which case the stability matrix becomes $\left\{(1-\alpha)\I + \alpha \Sm \right\}\D$, and we immediately see that a new modifier can invade a population near equilibrium if and only if it reduces $\alpha$ so as to produce $\rho( \M{(\alpha)} \D ) > 1$ (\citealt[Theorem 3.9, pp. 126-128]{Altenberg:1984}; \citealt[Result 3]{Altenberg:and:Feldman:1987}).  If $r > 0$, one is left to evaluate
\eb
\da{} \rho( [(1-\alpha) \M_1 + \alpha \M_2] \D ), 
\ee
where $\M_1, \M_2 \neq \I$.   A closed-form characterization of stochastic matrices $\M_1$ and $\M_2$ that produce $d \rho / d \alpha < 0$ is not readily obtained.

The case of general $r > 0$ would seem unlikely to reverse the reduction result, since the only action of recombination with the modifier locus is to blend the equilibrium distribution of selected haplotypes with the distribution created by the new modifier allele.  This blending should not alter the tendency of modifier alleles that  reduce the transformation rate to become associated with haplotypes of above-average fitness (and the converse), which is the essence of the dynamics.  The blending would merely lessen the association.  Even free recombination cannot completely eliminate this emergent association:  a modifier allele that produces perfect transmission always invades when introduced, even if unlinked to the selected loci, i.e. $\rho(\M(0, 1/2) \; \D) > 1$ (\citealt[Theorem 3.5, p.118]{Altenberg:1984}; \citealt[Result 2]{Altenberg:and:Feldman:1987}). 

Thus it was conjectured in \citet{Altenberg:and:Feldman:1987} that for arbitrary $r \in [0, 1/2]$, a new modifier allele can invade if and only if it decreases $\alpha$.  This paper proves the conjecture through an extension and expanded use of Karlin's theorem.  The proof requires the extension of Karlin's theorem to essentially non-negative matrices (ML matrices), and to reducible matrices.  The last of the constraints needed to prove the general reduction principle is thus removed.

I preface the result with a self-contained review of the general modifier model developed in  \citet{Altenberg:1984,Altenberg:and:Feldman:1987} and also used in \citet{Zhivotovsky:Feldman:and:Christiansen:1994}.  

\section{The Model}

The evolutionary model examined here fits the general form
\eb
\label{eq:GeneralModel}
\wb \; z_i' = \sum_{jk} {\cal T}(i\la j, k) \; w_{jk} \; z_j \; z_k,
\ee
where $z_i$ is the frequency of haplotype $i$ in the population, $z_i'$ is the frequency in the next generation, ${\cal T}(i \la j,k)$ is the probability that parental haplotypes $j$ and $k$ produce an offspring haplotype $i$, $w_{jk}=w_{kj}$ is the fitness of diploid genotype $jk$, and $\wb=\sum_{jk} w_{jk} \, z_j \, z_k$ is the mean fitness of the population.  The model includes the general assumptions of an infinite population, frequency-independent viability selection, random mating, sex symmetry, no sex linkage, and non-overlapping generations.

The modifier gene model is a special case of \eqref{eq:GeneralModel} in which the genome is structured to contain a group of loci under selection, and a neutral locus external to the group that modifies their genetic transmission probabilities.  The structure is illustrated in Fig. \ref{fig:ModifierModel}.  Haplotypes will now have two indices, one for the allele at the modifier locus ($a, b, c$ etc.), and one for the haplotype of the selected loci ($i, j, k$ etc.).  The modifier allele is assumed to be transmitted perfectly (no mutation nor segregation distortion), so that the only force acting upon it arises from its association with the selected loci.  
Recombination between the modifier locus and the nearest selected locus occurs at rate $r_{ab}$.  

\begin{figure}
\vspace*{.05in} 
\centerline{\includegraphics[width=4in]{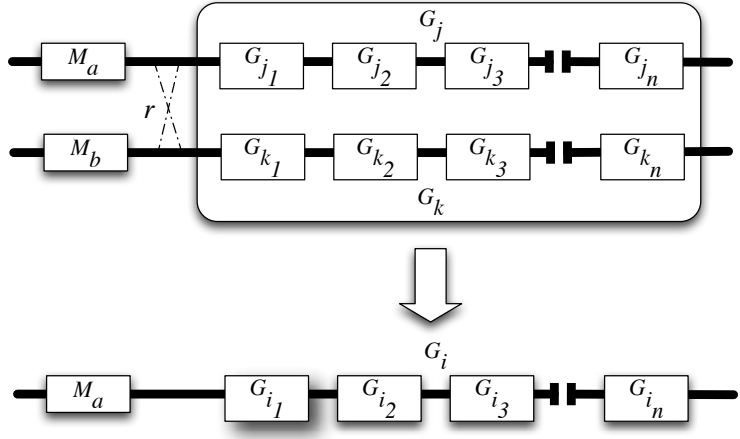}}
\caption{\label{fig:ModifierModel} The configuration of loci in the modifier gene model considered here.   $M_a$ and $M_b$ are the modifier alleles, $G_j$ and $G_k$ are the haplotypes undergoing viability selection, consisting of genes ${G_j}_1 \ldots {G_j}_n$ and ${G_k}_j \ldots {G_k}_n$.  With probabilities under the control of the modifier genotype $M_a/M_b$, haplotype $M_a G_i$ is produced.}
\end{figure}

The situation where the modifier locus is internal to the set of selected loci is not considered because it is no longer possible to separate recombination events with the modifier locus from transformations acting on the selected haplotype.  In order for the variation in  transmission to be linear in this configuration, there would need to be complete interference between recombination events on either side of the modifier locus and any other transformation processes being modeled.   

\subsection{The Modifier Gene Model}

Under this genetic structure, the transmission probabilities take on the form:
\[
T_{(r)}(ai \la aj | bk) := (1-r_{ab}) T(ai \la aj | bk) + r_{ab} T^{\Join}(ai \la ak | bj)
\]
where the probability that parental genotype $aj, bk$ produces gamete haplotype $ai$, conditioned on the offspring containing modifier allele $a$,  is:
\begin{description}
\item[$T(ai \la aj | bk)$,] when no recombination occurs between the modifier and nearest selected locus, and
\item[$T^{\Join}(ai \la ak | bj)$,] when recombination occurs between the modifier and nearest selected locus (hence $aj | bk$ becomes $ak | bj$).
\end{description}
So, $1 = \sum_i  T(ai \la aj | bk) = \sum_i  T^{\Join}(ai \la ak | bj) , \; \forall a, b, j, k$.

\subsection{Dynamical Recursion for the Modifier Model}
With this genetic structure, recursion \eqref{eq:GeneralModel} becomes:
\eb
\label{eq:ModifierModel}
\wb  \;  z_{ai}' =  \sum_{bjk} T_{(r)}(ai \la aj | bk) \; w_{jk}  \; z_{aj}  \; z_{bk}
\ee
where
\begin{description}
\item[$z_{ai}$]  is the frequency of the haplotype with allele $a$ at the modifier locus, and haplotype $i$ at the selected loci, and $z_{ai}'$ is for the next generation;
\item[$ w_{jk}$] is the fitness coefficient for diploid genotype $jk$ at the loci under selection;
\item[$\wb := \sum_{abjk}w_{jk}  \; z_{aj}  \; z_{bk}$] is the mean fitness of the population,
\end{description}

The evolutionary analysis consists of asking how the values of $T_{(r)}(ai \la aj | bk)$ determine whether a modifier allele $a$ can invade a population and be protected from extinction.  

\section{Equilibria and their Stability}
A population at equilibrium under \eqref{eq:ModifierModel} must satisfy the constraint:
\eb
\label{eq:Equilibrium}
\wbh \; \zh_{bi} =  \sum_{cjk} T_{(r)}(bi \la bj | ck) \; w_{jk}  \; \zh_{bj}  \; \zh_{ck}.
\ee
A perturbation of the equilibrium to $z_{bi} = \zh_{bi} + \epsilon_{bi}$ produces:
\eban
\lefteqn{\left(\wbh + 2 \sum_{bjck} \epsilon_{bj} w_{jk}  \zh_{ck}+ \sum_{bjck} \epsilon_{bj} \epsilon_{ck}\right)  \;   (\zh_{bi} + \epsilon_{bi}')} \\
&=&  \sum_{cjk} T_{(r)}(bi \la bj | ck) \; w_{jk}  \; (\zh_{bj} + \epsilon_{bj} ) \; (\zh_{ck} + \epsilon_{ck}).
\eean

The long term evolution of genetic transmission depends on the properties that allow a new modifier allele to invade a population and be protected from extinction.  Hence the analysis focuses on perturbations of the equilibrium by rare modifier alleles, entailing $\zh_{ai} = 0$ for all $i$ for new modifier allele $a$.  Making this substitution, and ignoring all second and higher order terms in the perturbation, the linear recursion on a new modifier allele, $a$, that perturbs \eqref{eq:Equilibrium} can be represented in vector form as:
\eb
\label{eq:ExternalStability}
\epsi_a' = \M \; \D \; \epsi_a
\ee
where $\M$ is a stochastic matrix, $\D$ a positive diagonal matrix, and 
\eban
&\epsi_a := \matrx{\epsilon_{ai}}_{i=1}^{n},  \mbox{\ \ }\D := \diag{\wh_i / \wbh }_{i,j=1}^{n}, \mbox{\ \ } \wh_i := \sum_{bj} w_{ij} \zh_{bj}, \mbox{\ and \ }&\\
&\M :=  \matrx{\sum_{bk} T_{(r)}(ai \la aj | bk) \frac{w_{jk}}{\wh_j} \zh_{bk}}_{i, j=1}^{n}&.
\eean
Modifier allele $a$ will increase at a geometric rate when rare if and only if the spectral radius $\rho(\M\D)$ exceeds 1.  Clearly, if $\D = \I$, then $\rho(\M\D) = \rho(\M) = 1$, so geometric rates of change in modifier allele frequencies require $\D \neq \I$, a situation described by saying there is a positive {\it selection potential} (\citealt[``fitness load'' p. 63]{Altenberg:1984}; \citealt{Altenberg:and:Feldman:1987}):
\eb\label{eq:SelectionPotential}
V =  \frac{\max_i \wh_i}{ \wbh} - 1 > 0.
\ee
The analysis consists of evaluating the relationship between  $T_{(r)}(ai \la aj | bk)$ and $\rho(\M\D)$.

\section{Variation in Transmission}
The reduction principle emerges in models where the modifier gene scales all the transition probabilities between different genotypes equally.  Variation in transmission falls along a line that intersects the matrix for perfect transmission (where parental haplotypes are transmitted unchanged to their gametes, and in equal proportions), which functions as an ``origin'' in this space of matrices.  For this reason it is called ``linear variation''  \citep{Altenberg:1984,Altenberg:and:Feldman:1987}.  

A mechanistic derivation of linear variation is that each selected haplotype $j$ has a certain probability, $\alpha$, of being ``hit'' by some transforming process, and given that it is hit, it is transformed into various other selected haplotypes with different probabilities. These probabilities, $T(i \la j | k)$, include possible dependence on both parental haplotypes $j$ and $k$. When the effect of the modifier gene is to scale the ``hit'' rate, $\alpha$, up or down equally for all haplotypes, it produces linear variation in transmission.

\subsection{Perfect Transmission}
A genetic system that perfectly transmits parental selected haplotypes to gametes can be represented by:
\eb
T_{\mbox{id}}(ai \la aj | bk) = T^{\Join}_{\mbox{id}}(ai \la aj | bk) =  \delta_{ij},
\ee
where $\delta_{ij} = 1$ if $i=j$ and $0$ otherwise.  Any genetic system can be characterized by the probability with which perfect transmission occurs.  For a given modifier genotype $ab$, a global characterization of the extent of perfect transmission can be captured by a lower bound, $1-\alpha_{ab}$, on the probability of perfect transmission over all selected haplotypes, where
\eb
1-\alpha_{ab} := \min_{i,k}\left\{ T(ai \la ai | bk),  T^{\Join}(ai \la ai | bk) \right\} \in [0, 1].
\ee
This lower bound can be used to parameterize the transmission probabilities:
\eb
\label{eq:Parameterized}
T(ai \la aj | bk) = (1-\alpha_{ab}) \; \delta_{ij} + \alpha_{ab} \;  P(ai \la aj | bk)
\ee
and
\eb
\label{eq:ParameterizedX}
T^{\Join}(ai \la aj | bk) = (1-\alpha_{ab}) \;  \delta_{ij} + \alpha_{ab} \;  P^{\Join}(ai \la aj | bk),
\ee
Thus $P(ai \la aj | bk) \geq 0$ and $P^{\Join}(ai \la aj | bk) \geq 0$  for all $i, j, k$, and for some $i,k$, $P(ai \la ai | bk) = 0$ or $P^{\Join}(ai \la ai | bk) = 0$.

{\bf Linear variation} is defined to be where the sole effect of the modifier locus is to determine $\alpha$:
\eb
\label{eq:TLinear}
T(ai \la aj | bk) = (1-\alpha_{ab}) \; \delta_{ij} + \alpha_{ab} \;  P(i \la j | k)
\ee
and
\eb
\label{eq:TLinearX}
T^{\Join}(ai \la ak | bj) = (1-\alpha_{ab}) \;  \delta_{ik} + \alpha_{ab} \;  P^{\Join}(i \la k | j).
\ee

\subsection{External Stability of a Modifier Monomorphism}  For the purposes of brevity, it becomes useful at this point to restrict the analysis to populations that are initially fixed on a single modifier allele.  The steps that follow apply to modifier polymorphisms, but require additional characterizations of the polymorphisms;  modifier polymorphisms under generalized transmission are analyzed in \citet{Altenberg:1984}, and for particular cases of transmission processes, in \citet{Liberman:and:Feldman:1986:MMR,Liberman:and:Feldman:1986:GRP,Feldman:and:Liberman:1986,Liberman:and:Feldman:1989}.

The population will initially be at equilibrium, fixed on modifier allele $b$ where $\alpha_{bb} > 0$.  Haplotypes will be listed only if they occur with positive frequency, so $\zh_{bi} > 0$ for all $i$ at equilibrium.  To this population at equilibrium, modifier allele $a$ is introduced.  Linear variation and $\alpha_{bb} > 0$ preclude the possibility that the new modifier allele will cause the production of selected haplotypes not already present in the population.

For readability, let $\alphah$ represent $\alpha_{bb}$, $\alpha$ represent $\alpha_{ab}$, and $r$ represent $r_{ab}$.  For linear variation \eqref{eq:TLinear}\eqref{eq:TLinearX}, the stability matrix $\M$ in \eqref{eq:ExternalStability} can be expressed as a function of $\alpha$ and $r$:
\eb
\begin{split}
\M(\alpha, r) &= \matrx{\sum_{k} T_{(r)}(ai \la aj | bk) \frac{w_{jk}}{\wh_j} \zh_{bk}}_{i, j=1}^{n} \\
&=  (1-\alpha) [(1\!-\!r) \I +  r \Q] \!\!\! + \alpha[ (1\!-\!r) \Sm +  r \Smt],
\end{split}
\ee
where
\eban
\lefteqn{\matrx{\sum_{k} \delta_{ij} \frac{w_{jk}}{\wh_j} \zh_{bk}}_{i, j=1}^{n} = \I, }\\
\Q &:=& \matrx{\sum_{k} \delta_{ik} \frac{w_{jk}}{\wh_j} \zh_{bk}}_{i, j=1}^{n} 
= \matrx{ \frac{w_{ij}}{\wh_j} \zh_{bi}}_{i, j=1}^{n},\\
\Sm  &:=&  \matrx{\sum_{k}   P(i \la j | k) \frac{w_{jk}}{\wh_j} \zh_{bk}}_{i, j=1}^{n}, \mbox{ and }\\
\Smt  &:=&  \matrx{\sum_{k}   P^{\Join} (i \la k | j) \frac{w_{jk}}{\wh_j} \zh_{bk}}_{i, j=1}^{n}.\eean
Letting $\ev$ be the vector of ones, and $\ev^\tp$ its transpose, note that
\eb
\label{eq:QSStMarkov}
\ev^\tp \Q = \ev^\tp \Sm = \ev^\tp \Smt = \ev^\tp,
\ee
since
$\sum_i P(i \la j|k) = \sum_i P^{\Join}(i \la k|j) = \sum_k {w_{jk}  \zh_{bk} / \wh_j} = 1$.
$\M(\alphah, r)$ can be used to express the equilibrium constraint \eqref{eq:Equilibrium}:
\eb
\label{eq:NewEquilibrium}
\begin{split}
\zvh_b &=  \M(\alphah, r) \D \zvh_b \\
&= \left\{ (1-\alphah) [(1\!-\!r) \I +  r \Q] + \alphah[ (1\!-\!r) \Sm +  r \Smt]\right\} \D \zvh_b. 
\end{split}
\ee

\section{Results}
The strategy for proving the reduction result for arbitrary $r$ is to try to allow the immediate application of Karlin's Theorem 5.2 by representing $\M(\alpha,r)$ as $\M(\alpha,r) = (1-\beta) \I + \beta \Mc$, for a choice of $\Mc$ where we know $\rho(\Mc\D)$.  Finding an irreducible stochastic such $\Mc$ would give 
\[
\rho(\M(\alpha,r) \D) > \rho(\Mc \D) \mbox{ for  } 0 < \beta < 1, 
\]
and
\[
\rho(\M(\alpha,r) \D) < \rho(\Mc \D)\mbox{ for  } \beta > 1.  
\]

For our choice of $\Mc$, we possess the crucial fact from \eqref{eq:NewEquilibrium} that $\rho(\M(\alphah,r) \D)=1$, so we might hope to find $\beta$ and $r$ to give $(1-\beta) \I + \beta \M(\alphah,r) = \M(\alpha,r)$, but this is unworkable.  Fortunately, there exists a matrix solution $\M(\alpha,r) = (1-\beta) \I + \beta \McSolve$, but $\McSolve$ may have negative elements along the diagonal for certain combinations of $\alpha$, $\alphah$, and $r$.  Such a matrix is known as a `Metzler-Leontief' \citep{Seneta:1981} or  `essentially nonnegative' \citep{Cohen:1981} matrix.  An extension of Karlin's Theorem 5.2 to ML-matrices is provided which allows proof for all $\alpha$, $\alphah$, and $r$.  We proceed in stages.

\subsection{Metzler-Leontief (ML) Matrices}
An ML-matrix is a square real matrix where all non-diagonal elements are non-negative.  The {\em spectral abscissa} of a matrix $\A$ is defined as $\pi(\A) = \max_i \{ \RealPart( \lambda_i)\}$, where $\{  \lambda_i \}$ are the eigenvalues of $\A$.  Let $\B$ refer to an ML-matrix, where $b_{ij} \geq 0$ for all $i\neq j$.   We utilize the following properties of ML-matrices:

\begin{Lemma}[Spectral Abscissa of ML Matrices]
\label{Lemma:ML}\ 
\benu
\item \label{ML:Abscissa} $\pi(\B)$ is an eigenvalue of $\B$, referred to as the Perron root as it is for positive matrices; 
\item \label{ML:BplusgI} $\pi(\B) = \rho(\B + g \I) - g$, for any $g \geq - \min_i \{ b_{ii}\}$; 
\item  \label{ML:PositiveLeftRight} If $\B \z = \lambda \B$ and $\z > 0$, then $\lambda = \pi(\B)$.  If in addition $\B$ is irreducible, then there also exists $\vv > 0$ such that $\vv^\tp \B = \pi(\B) \vv^\tp$, and $\z$ and $\vv$ are unique up to constant  multiples.
\eenu
\end{Lemma}
\begin{proof}
These are found in, or follow directly from, Theorem 2.6 in \citep[pp. 45--46]{Seneta:1981} and Theorem 3 in \citep[p. 66]{Gantmacher:1959vol2}. 
\end{proof}
\subsection{Solving for {\normalsize $\McSolve$}}

\begin{Lemma}
\label{Lemma:M*}
Define
\[
\Mc(\alpha, c_1, c_2) := (1-\alpha) \left[ \left( 1-c_1 \right) \I + c_1 \Q \right]  + \alpha   [ (1-c_2) \Sm + c_2 \Smt ]. 
\]
Then for $\alpha, \alphah \in (0, 1)$,  and $r \in [0,1]$:
\benu
\item $\M(\alpha,r) = (1-\beta) \I + \beta \Mc(\alphah, c_1, c_2)$ is solved by 
\[
\beta = \alpha / \alphah, \ c_1 = r \alphah(1-\alpha) / [\alpha(1-\alphah)], \mbox{\ and \ } c_2 = r,
\]
to give
\eb
\label{eq:Malphar}
\M(\alpha, r) = \left(1-\frac{\alpha}{\alphah}\right) \I + \frac{\alpha}{\alphah} \McSolve,
\ee
where
\eba
\label{eq:M*}
\dsty \McSolve &:=&  \Mc\left(\alphah, r \frac{\alphah(1-\alpha)}{\alpha(1-\alphah)}, \; r \right) \\
&=&
\left(1 \! -\alphah\! -\! r \frac{\dsty  \alphah }{\dsty\alpha} \!+ \!r \alphah \right) \I + r \alphah \left(\frac{1}{\dsty \alpha}-1\right) \Q + \alphah [ (1-r) \Sm + r \Smt]. \notag
\eea
\item $\ev^\tp \Mc(\alphah, c_1, r) = \ev^\tp$ and $\Mc(\alphah, c_1, r) \,  \D \, \zvh_b = \zvh_b$ for any $c_1 \in \Re$.
\item $\McSolve$ is an ML-matrix and, if $(1-\alphah) r \Q + \alphah[(1-r) \Sm + r \Smt]$ is irreducible, $\pi\left( \McSolve \, \D \right) = 1$.   
\item If, in addition, $\alpha \geq \alphah$ or $\alpha \geq  \alphah \; r / [1 - \alphah (1-r)]$, then $\McSolve$ is non-negative and $\pi\left( \McSolve \, \D \right) = \rho\left( \McSolve \, \D \right) = 1$.
\eenu
\end{Lemma}
\begin{proof}
Straightforward evaluation verifies \eqref{eq:Malphar} and \eqref{eq:M*} from:
\eban
\McSolve &=& (1-\alphah) \left[ \left( 1- r\frac{\alphah(1-\alpha)}{\alpha(1-\alphah)}\right) \I +  r\frac{\alphah(1-\alpha)}{\alpha(1-\alphah)} \Q \right] \\
&& + \  \alphah  \;\;  [ (1-r) \Sm + r \Smt ].
\eean
By \eqref{eq:QSStMarkov}, $\ev^\tp \Mc(\alphah, c_1, c_2) = [(1-\alphah)(1-c_1+c_1) + \alphah(1-c_2+c_2)]\ev^\tp = \ev^\tp$.  Observe that $\I \D \zvh_b = \Q \D\zvh_b = \D \zvh_b$:
\eban
\Q \D \zvh_b &=&  \matrx{ \frac{w_{ij}}{\wh_j} \zh_{bi}}_{i, j=1}^{n} \diag{\frac{\wh_i}{\wbh}}_{i,j=1}^{n}  \zvh_b \\
&=&  \matrx{\sum_j \frac{w_{ij}}{\wbh} \zh_{bi} \zh_{bj}}_{i,j=1}^{n} =  \matrx{\frac{\wh_{i}}{\wbh} \zh_{bi}}_{i,j=1}^{n} = \D \zvh_b.
\eean
Hence $\Mc(\alphah, c_1, r) \,  \D \, \zvh_b =$
\[
(1-\alphah) \left[ \left( 1 \! - \! c_1 \right) \I + c_1 \Q \right]  + \alphah   [ (1-r) \Sm + r \Smt ] \D \zvh_b
\]
is invariant for all $c_1 \in \Re$.  Since $\Mc(\alphah, r, r) := \M(\alphah,r)$, then
\eb
\label{eq:M*DVz=z}
\Mc(\alphah,  c_1, r)  \D \zvh_b = \Mc(\alphah, r, r)  \D \zvh_b = \M(\alpha,r) \D \zvh_b = \zvh_b.
\ee

The off-diagonal elements in \eqref{eq:M*} are non-negative as they all derive from non-negative stochastic matrices---$\Q$, $\Sm$, and $\Smt$---multiplied by non-negative coefficients under conditions $\alphah \in (0,1)$, $c_1 \geq 0$, and $r \in [0,1]$.     

If $(1-\alphah) r \Q + \alphah[(1-r) \Sm + r \Smt]$ is irreducible, then since $\zvh_b$ is strictly positive, $\zvh_b$ is the right Perron eigenvector of $\Mc(\alphah,  c_1, r)  \D$ with Perron root 1, by Lemma \ref{Lemma:ML} (\ref{ML:PositiveLeftRight}).  The coefficient on $\I$, $1 -\alphah  -  r \alphah  / \alpha+ r \alphah$, is non-negative when $\alpha > \alphah$, for then $1 -\alphah -  r \alphah / {\alpha}+r \alphah \geq (1-\alphah)(1-r) \geq 0$, but is negative if $\alpha <  \alphah \; r / [1 - \alphah (1-r)]$ (found by simple rearrangement).  Thence $\McSolve$ may have negative diagonal elements.  When $\McSolve$ is non-negative, then $\pi(\McSolve \, \D) = \rho(\McSolve \, \D)$. 
\end{proof}
\subsection{Extending Karlin's Theorem 5.2 to ML-Matrices}
\ \\
Friedland's \citeyearpar{Friedland:1981} Donsker-Varadhan related variational formula for the spectral radius is applied to $\A \geq 0$ in order to show that it applies also to the Perron root of $\B = \A - g \, \I$:
\begin{Lemma}[Variational Formula for ML-Matrices]
\label{Lemma:muB}
For any irreducible ML-matrix $\B$, the spectral abscissa is
\eb
\label{eq:muB}
\pi ( \B ) =  \sup_{\pv \in \Pc_n} \inf_{\x > \0} \sum_{i=1}^n p_i  \frac{[\B \x]_i}{x_i},
\ee
where $\Pc_n = \{ \pv \colon \pv \geq 0, \ev^\tp \pv = 1\}$.
\end{Lemma}
\begin{proof}
Define for any matrix $\A$, 
\eb
\label{eq:DVfunction}
f(\A, \pv, \x) := \sum_{i=1}^n p_i  \frac{[ \A \x]_i}{x_i}.
\ee
The Donsker-Varadhan related variational formula \citep[Corollary 3.1]{Friedland:1981} for irreducible $\A \geq 0$ gives us $\rho(\A)  = \sup_{\pv \in \Pc_n} \inf_{\x > \0} f(\A, \pv, \x)$.  Substituting $\A = \B + g \, \I$:
\eban
\rho(\A) & =& \rho( \B + g \, \I ) 
= \sup_{\pv \in \Pc_n} \inf_{\x > \0} \sum_{i=1}^n p_i  \frac{[ ( \B + g \, \I ) \x]_i}{x_i}  \\
& =& \sup_{\pv \in \Pc_n} \inf_{\x > \0} \sum_{i=1}^n p_i  \frac{[ \B \x]_i}{x_i} \; \; + g.
\eean
Knowing from Lemma \ref{Lemma:ML} (\ref{ML:BplusgI}) that $\rho(\A ) = \pi( \B) + g$ proves the lemma.  Moreover, 
\[
\sup_{\pv \in \Pc_n} \inf_{\x > \0} f(\A, \pv, \x)
\]
and
\[
\sup_{\pv \in \Pc_n} \inf_{\x > \0} f(\B, \pv, \x)
\]
are both attained for the same $\pv$ and $\x$. 
\end{proof}

\begin{Lemma}[Derivative of the Spectral Abscissa]
\label{Lemma:dbB}
Let $\B(\beta)$ be a function of $\beta \in \Re$, in continuity class $C^2$, such that $\B(\beta)$ is an ML-matrix for $\beta \geq 0$.  
Let $\pv(\beta)$ and $\x(\beta)$ ($\x$ normalized so $\ev^\tp \x(\beta) = 1$) be the vectors at which the supremum and the infimum in \eqref{eq:muB}, are attained, respectively.  Then
\eban
\db{} \pi ( \B(\beta) ) = \sum_{i=1}^n p_i(\beta)  \frac{[ \dsty \db{\B(\beta)} \;\; \x(\beta)]_i}{x_i(\beta)}.
\eean
\end{Lemma}
\begin{proof}
Let $\A(\beta) = \B(\beta) + g \, \I \geq 0$ be a non-negative matrix associated with $\B(\beta)$.    So $\db{} \A(\beta) = \db{}\B(\beta)$.  By Lemma \ref{Lemma:ML} (\ref{ML:BplusgI}), $\db{} \rho(\A(\beta) ) = \db{} \pi(\B(\beta))$.  Differentiating \eqref{eq:DVfunction}:  $\db{} f(\A, \pv, \x) = $
\[
\frac{\partial f(\A, \pv, \x)}{\partial \A} \db{\A} + \frac{\partial f(\A, \pv, \x)}{\partial \pv} \db{\pv} + \frac{\partial f(\A, \pv, \x)}{\partial \x}\db{\x}
\]
Since  $\pv(\beta)$ and $\x(\beta)$ are unique critical points of $f(\A(\beta), \pv(\beta), \x(\beta))$ \citep{Friedland:and:Karlin:1975,Friedland:1981,Karlin:1982}:
\eb\nonumber
\label{eq:CriticalPts}
\left. \frac{\partial f(\A, \pv, \x)}{\partial \pv} \db{\pv} \right |_{\Above{\A (\beta),}{ \pv(\beta), \x(\beta)}} = \left.  \frac{\partial f(\A, \pv, \x)}{\partial \x}\db{\x} \right|_{\Above{\A (\beta),}{ \pv(\beta), \x(\beta)}} = 0,
\ee
hence
\eban
\lefteqn{\db{} \rho(\A(\beta) ) = \db{} \pi(\B(\beta)) =  \left. \frac{\partial f(\A, \pv, \x)}{\partial \A} \db{\A}\right|_{\Above{\A(\beta),}{\pv(\beta), \x(\beta)} } }\\[8pt]
&=& 
\sum_{i=1}^n p_i(\beta)  \frac{[\db{\dsty \A(\beta)} \x(\beta)]_i}{x_i(\beta)}
= \sum_{i=1}^n p_i(\beta)  \frac{[\db{\dsty \B(\beta)} \x(\beta)]_i}{x_i(\beta)}.
\eean
\end{proof}

\begin{Theorem} [Extension to ML-Matrices]
\label{Theorem:KarlinForML}
Let $\Lm$ be an irreducible ML-matrix such that $\ev^\tp \Lm = \ev^\tp$.  Consider the family of matrices
 \[
 \C(\beta) = (1-\beta) \I + \beta \, \Lm
 \]
with $\beta > 0$.  Then for any diagonal matrix $\D$ not a multiple of $\I$, with positive terms on the diagonal, the spectral abscissa and Perron root
\[
\pi( \C{(\beta)} \D )
\]
is strictly decreasing as $\beta $ increases.
\end{Theorem}
\begin{proof}
Let $\B(\beta) = \C(\beta) \D$.  Applying Lemma \ref{Lemma:dbB} we have:
\eban
\db{} \pi (  \C(\beta) \D ) = \sum_{i=1}^n p_i(\beta)  \frac{[(\Lm- \I)  \D  \x(\beta)]_i}{x_i(\beta)},  
\eean
since $\db{}\B(\beta) =  \db{}\C(\beta) \; \D = (\Lm- \I)  \D$.

Following Karlin's proof in \citet{Karlin:1982}, we note that for $\beta > 0$, 
\[
\Lm - \I = \frac{1}{\beta} \left([(1-\beta)\I + \beta \Lm]  - \I \right) = \frac{1}{\beta}[\C(\beta) - \I ].
\]
Hence
\eba
\label{eq:dmuCalc}
\db{ \pi ( \C(\beta) \D )} &=& \frac{1}{\beta} \sum_{i=1}^n p_i(\beta)  \frac{[(\C(\beta) - \I)  \D  \x(\beta)]_i}{x_i(\beta)} =  
\dsty  \frac{1}{\beta} \!\!\left[ \pi ( \C(\beta) \D ) \! -  \!\!\sum_{i=1}^n p_i(\beta)   \frac{\D_{ii} x_i}{x_i} \right] \notag \\ 
&=& \frac{1}{\beta} \left[ \pi (\C(\beta) \D )\!  - \!\! \sum_{i=1}^n p_i(\beta)  \D_{ii} \right].
\eea
As $\x(\beta)$ is unique and produces the infimum, for $\0 < \x \neq \x(\beta) $:
\eb\nonumber
\label{eq:xbetax}
\sum_{i=1}^n p_i(\beta)  \frac{[ \C(\beta) \D \x(\beta)]_i}{x_i(\beta)} <
\sum_{i=1}^n p_i(\beta)  \frac{[ \C(\beta) \D  \x]_i}{x_i}.
\ee

Let $\y > \0$ be the right Perron eigenvector of $\C(\beta)$.  Since $\ev^\tp \Lm =  \ev^\tp \C(\beta) = \ev^\tp$, we know that $\pi(\C(\beta))=1$, hence $\C(\beta) \; \y = \y$.  Now, set $\x = \D^{-1} \y$, where $\y$ is scaled so that $\ev^\tp \x = 1$.  Assuming $\D^{-1} \y \neq \x(\beta)$, we get:
\eba 
\label{eq:KarlinForMLKey}
\lefteqn{\pi(\C(\beta)\D) =
\sum_{i=1}^n p_i(\beta)   \frac{[ \C(\beta) \D \x(\beta)]_i}{x_i(\beta)}
} \\
&<  & \sum_{i=1}^n p_i(\beta)  \frac{[ \C(\beta) \D  \D^{-1} \y]_i}{(\D^{-1} \y)_i} =  
\sum_{i=1}^n p_i(\beta)  \frac{y_i}{ y_i / \D_{ii}}=\sum_{i=1}^n p_i(\beta) \D_{ii} \nonumber
\eea
We verify that $\D^{-1} \y \neq \x(\beta)$: should $\D^{-1} \y = \x(\beta)$, then \eqref{eq:KarlinForMLKey} is an equality, thence $\pi ( \C(\beta) \D ) =  \sup_{\pv \in \Pc_n}  p_i \D_{ii} = \max_i  \D_{ii} $.  The supremum requires $p_i(\beta)=0$ for all $\{i \colon \; \D_{ii} < \max_j  \D_{jj} \}$, which is nonempty because $\D \neq d \, \I$ for any $d$.  In contradiction, irreducible $\Lm$ implies $\pv > \0$ \citep[{eq. (3.5)}, Theorem 3.2, and Corollary 3.1]{Friedland:1981}.

Application to \eqref{eq:dmuCalc} gives:
\[
\db{} \pi ( \C(\beta) \D )  =  \frac{1}{\beta} \left( \pi (\C(\beta) \D ) -  \sum_{i=1}^n p_i (\beta)  \D_{ii} \right) < 0. 
\] 
\end{proof} 

\begin{Corollary}[Extension to Substochastic-ML Matrices]
\label{Corollary:KarlinForSubstochasticML}
Let $\Lm$ be an irreducible ML-matrix such that $\ev^\tp \Lm \leq \ev^\tp$ and $\ev^\tp \Lm \neq \ev^\tp$.  By extension from non-negative matrices this will be referred to as a substochastic-ML matrix.  Consider the family of matrices $ \C(\beta) = (1-\beta) \I + \beta \, \Lm$ with $\beta > 0$.  Then for any diagonal matrix $\D$ with positive terms on the diagonal, the Perron root and spectral abscissa $\pi( \C{(\beta)} \D )$ is strictly decreasing as $\beta$ increases.
\end{Corollary}
\begin{proof}
The steps are identical to those in the proof of Theorem \ref{Theorem:KarlinForML} except that $\ev^\tp \Lm \leq \neq \ev^\tp$ means $\ev^\tp \C(\beta) \leq \neq \ev^\tp$, hence $\pi(\C(\beta)) < 1$ (\citealt[Corollary 3, p. 30; Corollary 3, p. 52]{Seneta:1981} and Lemma \ref{Lemma:ML} (\ref{ML:BplusgI})).  No assumption that $\D^{-1} \y \neq \x(\beta)$ is needed, and \eqref{eq:KarlinForMLKey} becomes:
\eban
\pi(\C(\beta)\D)& \leq  
&\sum_{i=1}^n p_i(\beta)  \frac{\pi(\C(\beta)) \, y_i}{(\D^{-1} \y)_i}
= \pi(\C(\beta)) \sum_{i=1}^n p_i(\beta) d_i  < \sum_{i=1}^n p_i(\beta) d_i .
\eean
Hence from \eqref{eq:dmuCalc}, $\db{} \pi(\C(\beta) \D) < 0$.  No assumption $\D \neq d\, \I$ is used.
\end{proof}

\subsection{Main Result}

\begin{Theorem}[The Reduction Principle for Linear Variation]
\label{Theorem:Main}
Suppose a new allele at a modifier locus is introduced into a population fixed at the modifier locus near an equilibrium with a positive selection potential among the selected haplotypes, and $\alphah < 1$. The new modifier allele produces linear variation in the transmission of the selected loci, and is linked to the nearest selected locus with recombination rate $r$.  

If $\M(\alpha, r)$ is irreducible, the new modifier allele will increase (decrease) in frequency at a geometric rate if it brings the transmission closer to (further away from) perfect transmission, i.e., $\alpha < \alphah$ ($\alpha > \alphah$).

If $\M(\alpha, r)$ is reducible, then:  
\benu
\item \label{ThmMain:Increase} the new modifier allele will increase in frequency at a geometric rate if $\alpha < \alphah$, provided it occurs in at least one haplotype whose marginal fitness differs from the mean fitness of the population; and if $\alpha > \alphah$, it will either: 
\item \label{ThmMain:Decrease} decrease in frequency at a geometric rate; or 
\item \label{ThmMain:NoIncrease} not increase in the case where it occurs in a haplotype whose block in the Frobenius normal form of $\M(\alpha, r)$ is an isolated block with all marginal fitnesses equal.
\eenu
\end{Theorem}
\begin{proof}
For $\alpha=0$, the theorem reduces to Result 2 in \citet{Altenberg:and:Feldman:1987}.  Henceforth only $\alpha > 0$ is considered.  Lemma \ref{Lemma:M*} shows that $\McSolve$ is an ML-matrix for all $\alpha, \alphah \in (0, 1)$,  and $r \in [0,1]$.  We substitute $\C(\beta) = (1-\beta) \I + \beta \McSolve$ and $\beta=\alpha/\alphah$, thus $\C(\alpha/\alphah) = \M(\alpha,r)$.

{\bf Case 1: Irreducible $\M(\alpha,r)\D$.}
  
If $\M(\alpha,r)$ is irreducible and $\D$ positive on the diagonal, $\M(\alpha,r)\D$ is irreducible, and Theorem \ref{Theorem:KarlinForML} applies.  Hence for $\beta > 0$, $\db{} \pi( \C(\beta) \D )  < 0$.  Therefore, 
\[
\pi( \C(\alpha/\alphah) \D) > \pi( \C(1) \D ) = 1, \mbox{ for } \alpha < \alphah,
\]
and 
\[
\pi( \C(\alpha/\alphah) \D) < \pi( \C(1) \D ) = 1, \mbox{ for } \alpha > \alphah.
\]
Since $\M(\alpha,r)$ is always non-negative, the Perron root is in fact the spectral radius: 
\[
\pi( \C(\alpha/\alphah) \D) = \pi(\M(\alpha,r)\D) =  \rho(\M(\alpha,r)\D).
\]
For $\alphah = 1 > \alpha$, $\McSolve$ is unsolvable, but as a limit for $\alphah < 1$, $\rho(\M(\alpha,r)\D) \geq 1$ is assured.

{\bf Case 2: Reducible $\M(\alpha,r)\D$.}

If $\M(\alpha,r)$ is reducible or $\D$ is not strictly positive on the diagonal, $\M(\alpha,r)\D$ is reducible, and Theorem \ref{Theorem:KarlinForML} must be extended.  To analyze the dependence of the spectral abscissa on $\alpha$ in the case of reducible $\C(\beta)$ or where some $\D_{ii}=0$, we utilize the Frobenius normal form of $\C(\beta) \D$.   The Frobenius normal form permutes the indices of $\C(\beta) \D$ so that along the diagonal are irreducible square block matrices $\C_h(\beta) \D_h$, where $h$ indexes the diagonal blocks.  

The Frobenius normal form, $\A$, of a reducible matrix has the structure \citep[p. 75]{Gantmacher:1959vol2}:
\eb\label{eq:FrobeniusNormalForm}
\A = \left(
\begin{array}{cccc|cccc}
\A_1 & \0 & \cdots & \0 & & && \\
\0 & \A_2 & \ddots & \vdots & &  \mbox{\Large \bf 0} & &\\
\vdots &  & \ddots & \0 & &  & &\\
\0 & \cdots & \0& \A_t & & & &\\
\hline \A_{t+1,1} & \A_{t+1, 2} & \cdots & \A_{t+1,t} & \A_{t+1}  & \0 &\cdots& \0\\
\vdots & \vdots & \cdots & \cdots & \cdots & \ddots & \ddots & \0\\
\A_{t+s,1} & \A_{t+s, 2} & \cdots & \A_{t+s, t} & \A_{t+s, t+1} &\cdots & \cdots &\A_{t+s}
\end{array}
\right)
\ee
Blocks $\A_1, \ldots, \A_t$ are referred to as {\em isolated} blocks, while  $\A_{t+1}, \ldots, \A_{t+s}$ are referred to as {\em non-isolated} blocks.   Isolation is defined in terms of the rows, so in a matrix of the form
\[
\left(
\begin{array}{cc}
\A_1 & \0  \\
\A_3 & \A_2 
\end{array}\right)
\]
the block $\A_1$ is defined as isolated, whereas $\A_2$ is non-isolated if $\A_3 \neq \0$.  Thus, by definition, block $\A_1$ in Frobenius normal form is always isolated.

At this point we require the following lemma:
\begin{Lemma}
The eigenvalues of $\A$ are the eigenvalues of the irreducible diagonal block matrices in the Frobenius normal form of $\A$.
\end{Lemma}
\begin{proof}
Let $h \in \Ic$ refers to isolated blocks, and $h \in \Nc$ refers to non-isolated blocks.  The non-isolated blocks are contained in a principal submatrix, $\A_\Nc$, of $\A$.  In \eqref{eq:FrobeniusNormalForm} this would be
\eb\label{eq:ANc}
\A_\Nc = \begin{pmatrix}
\A_{t+1}  & \0 &\cdots& \0\\
 \cdots & \ddots & \ddots & \0\\
 \A_{t+s, t+1} &\cdots & \cdots &\A_{t+s}
\end{pmatrix}.
\ee

The eigenvalues of $\A$ must be eigenvalues either of an isolated block $\A_h$, $h \in \Ic$, or of the submatrix of non-isolated blocks, $\A_\Nc$.  This is readily seen, because, letting $\x$ be an eigenvector of $\A$, if $\A \x = \lambda \ \x$, then for isolated block $h$, $\A_h \x_h = \lambda \ \x_h$, where $\x_h$ refers to the subvector of $\x$ with indices in block $h$.  Hence either 
\benu
\item $\lambda$ is an eigenvalue of $\A_h$; or
\item $\x_h = \0$.   If $\x_h = \0$ for all $h \in \Ic$, then $\A_\Nc \x_\Nc = \lambda \ \x_\Nc$, hence $\lambda$ is an eigenvalue of $\A_\Nc$.
\eenu

In submatrix $\A_\Nc$, we see from \eqref{eq:ANc} that $\A_{t+1}$ is an isolated block.  Hence, repeating the argument above, the eigenvalues of $\A_\Nc$ must be eigenvalues of either $\A_{t+1}$ or of the principal submatrix containing blocks $\A_{t+2}$ through $\A_{t+x}$.  Continuing recursively in this manner, the lemma is proved.
\end{proof}

Thus the eigenvalues of  $\C(\beta)\D$ --- including the spectral abscissa --- must be eigenvalues of the diagonal blocks $\C_h(\beta)\D_h$.  Therefore, the spectral abscissa for $\C(\beta)\D$ is the maximum of the spectral abscissas:
\[
\pi( \C(\beta)\D ) = \max_{h} \pi( \C_h(\beta)\D_h ).
\]

The spectral abscissas of each block are either constant, or strictly decreasing:
	\benu
	\item Constant if $\D_h = d_h \; \I_h$ and $\ev_h^\tp \C_h(\beta) = \ev_h^\tp$;
	\item Strictly decreasing if either
		\benu
 		 \item $\D_h \neq d\  \I_h$ (by Theorem \ref{Theorem:KarlinForML}), or \\[-8 pt]
 		 \item $\C_h(\beta)$ is substochastic-ML, i.e. $\ev_h^\tp \C_h(\beta) \leq \ev_h^\tp, \neq  \ev_h^\tp$ (by Corollary \ref{Corollary:KarlinForSubstochasticML}).  
		\eenu
	\eenu
	
In the case where all blocks are isolated, it is not possible for all blocks to be constant, as this would entail that $\D = \I$, which is contrary to hypothesis \eqref{eq:SelectionPotential}.  In the case where there is at least one non-isolated block, then (referring to the indices in \eqref{eq:FrobeniusNormalForm}) there is at least one non-zero matrix $\C_{t+k, j}$, for some $t \in \{1, \ldots, s\}$, and $j \in \{1, \ldots, t\}$.  This entails that
\[
\ev_{t+k}^\tp  \C_{t+k, j}(\beta) + \ev_{j}^\tp \C_j(\beta) \leq \ev_{ j}^\tp, \neq  \ev_{ j}^\tp
\]
hence $ \ev_{j}^\tp \C_j(\beta) \leq \ev_{ j}^\tp, \neq \ev_{ j}^\tp$, thus $\C_j(\beta)$ is a substochastic-ML matrix.  Hence if $\Nc$ is nonempty, at least one isolated block will be an irreducible substochastic-ML matrix, which by Corollary \ref{Corollary:KarlinForSubstochasticML} is a decreasing block.  Therefore, in all cases, there will be at least one isolated block $h$ for which $\C_h(\beta) \D_h$ is strictly decreasing in $\beta$.

We have as a reference point \eqref{eq:NewEquilibrium}, where for $\beta = 1$, $\C(1)\D \zvh_b = \zvh_b > \0$.  Hence, by Lemma \ref{Lemma:ML} (\ref{ML:PositiveLeftRight}), $\pi(\C(1) \D) = 1$.  We know by Theorem 6 in \citep[pp. 77--78]{Gantmacher:1959vol2}, extended to ML-matrices by Lemma \ref{Lemma:ML} (\ref{ML:PositiveLeftRight}), that $\pi(\C_h(1) \D_h) = 1$ for all $h \in \Ic$, and $\pi(\C_h(1) \D_h) < 1$ for all $h \in \Nc$.  

For the parameter range $0 < \beta = \alpha/\alphah < 1$, i.e. $\alpha < \alphah$, at least one isolated block has spectral abscissa strictly decreasing in $\beta$, thus $\pi( \C_h(\beta)\D_h )  > 1$ for $\beta < 1$.  Hence,
\eb\label{eq:BetaLT1}
\pi( \C(\alpha / \alphah )\D ) = \rho(\M(\alpha, r) \D ) = \max_{h} \pi( \C_h(\alpha/\alphah)\D_h )  > 1,  \mbox{\ for \ } \alpha < \alphah.
\ee
As long as the new modifier allele $a$ is introduced into haplotypes that are part of the isolated block for which $\pi(\C_h(\alpha/\alphah) \D_h) > 1$, i.e $\epsi_{a_h} \neq \0$, then $\epsi_a$ will grow at a geometric rate.  This proves conclusion (\ref{ThmMain:Increase}) in the theorem.

For the parameter range $\beta = \alpha/\alphah > 1$, i.e. $\alpha > \alphah$, if there are no isolated blocks for which $\D_h = d_h \; \I_h$, then $\pi( \C_h(\beta)\D_h )  < 1$, for all $h \in \Ic$ when $\beta < 1$.  And recalling that $\pi(\C_h(1) \D_h) < 1$ for all $h \in \Nc$, then $\pi(\C_h(\beta) \D_h) < 1$ for all $h \in \Nc$ when $\beta > 1$ since all blocks are either constant or decreasing.  Therefore,
\[
\pi(\C (\alpha / \alphah ) \D) = \rho(\M(\alpha, r) \D ) = \max_{h} \pi( \C_h(\beta)\D_h )  < 1,  \mbox{\ for \ } \beta > 1,
\]
yielding conclusion (\ref{ThmMain:Decrease}) in the theorem.  If there is a constant isolated block, which requires $\D_h = d_h \; \I_h$, then $\pi(\C(\beta) \D) = \rho(\M(\alpha, r) \D ) = 1$ for $\beta > 1$.  If the new modifier allele is introduced associated with haplotypes within the constant block, then conclusion (\ref{ThmMain:NoIncrease}) in the theorem holds. 
\end{proof}
Conditions for $\M(\alpha,r)$ to be reducible are limited but important.  For $r=0$, $\Sm$ must be reducible.  This is the case in general when recombination is the only transformation acting on the selected haplotypes, and the Frobenius normal form of $\Sm$ consists of isolated irreducible blocks.  If no block has all marginal fitnesses the same, then conclusion (2) applies.  For $\M(\alpha,r)$ to be reducible when $r > 0$ and $\alpha < 1$, $\Q + \Sm + \Smt$ must be reducible.  This is precluded if all fitnesses $w_{ij}$ are positive, making $\Q > \0$.  For $\Q + \Sm + \Smt$ to be reducible requires either that some haplotypes be lethal, or that $\Sm + \Smt$ be reducible and that specific zeros in $\Sm + \Smt$ be matched by $w_{ij} = 0$, a highly non-generic selection regime. 

\section{Discussion}
 
\subsection{Mathematical Issues}
Theorem \ref{Theorem:Main} establishes the full generality of the reduction principle for modifier genes that produce linear variation in genetic transmission, assuring that no exceptions to the evolutionary reduction of genetic transformation rates can come about due to selection regime, allele or locus multiplicity, different genetic processes, or recombination with the modifier locus, for modifiers that produce linear variation in transmission.  

In previous studies, analytical tractability required tradeoffs between various simplifying assumptions:
  \citet{Liberman:and:Feldman:1986:MMR,Liberman:and:Feldman:1986:GRP,Feldman:and:Liberman:1986,Liberman:and:Feldman:1989} allowed no more than two alleles at each selected locus, or two demes in the case of migration modification, in order to obtain results for all $r > 0$.   \citet{Altenberg:1984,Altenberg:and:Feldman:1987} required $r=0$ or $\alpha=0$ in order to obtain results for arbitrary allele and locus number, genetic processes, and selection regimes.  These tradeoff are now removed.

In the process of proving Theorem \ref{Theorem:Main}, I have had to extend Karlin's Theorem 5.2 to ML-matrices, substochastic matrices, and reducible matrices, which may have application outside the current context of modifier gene theory.  Moreover, the treatment of reducible matrices here produces much stronger and more detailed results than in \citet{Altenberg:1984,Altenberg:and:Feldman:1987}, where the spectral radius $\rho(\M(\alpha, 0) \D$ was only proven to be non-increasing in $\beta$ for reducible $\M(\alpha,0)$.  Here, it is proven that $\rho(\M(\alpha, r) \D$ is strictly decreasing on $\alpha \in (0, \alphah)$, and strictly decreasing on $\alpha \in (\alphah, 1)$ except when the marginal fitnesses are all equal within an isolated block in the Frobenius normal form of $\M(\alpha, r) \D$, in which case $\rho(\M(\alpha, r) \D = 1$.

It should be noted that the method of proof utilized here does not allow us to say that $\da{}\rho(\M(\alpha, r) \D) < 0$ for $r > 0$.  That would be a much more direct means of proof of Theorem \ref{Theorem:Main} and would obviate the need to extend Karlin's theorem to ML-matrices.  A proof, however, runs up against impediments at step \eqref{eq:KarlinForMLKey} that have not been surmounted.

Such a result would be necessary to say, as was done in Result 2 in \citet{Altenberg:and:Feldman:1987}, that the asymptotic strength of selection either for or against a new modifier allele increases with $|\alpha - \alphah|$.  Theorem \ref{Theorem:Main} alone cannot not rule out the possibility that the function $\rho(\M(\alpha, r)$ meanders non-monotonically over $(0, \alphah)$ or $(\alphah, 1)$, only that it never crosses the $\rho=1$ line on those intervals.  This perverse possibility is not consistent, however, with the general finding that spectral functions are convex \citep{Friedland:1981}, and out of parsimony considerations, I would conjecture that for any $r > 0$, $\da{}\rho(\M(\alpha, r) \D) < 0$.  

\subsection{Linear Variation}
Linear variation is in a sense `impartial' in that all transitions --- from advantageous to deleterious haplotypes, and vice versa --- are scaled exactly the same.  So the production of both advantageous as well as deleterious haplotypes is reduced by a modifier that reduces transformation rates. 
But Theorem \ref{Theorem:Main} shows that a reducing modifier allele always generates its own linkage disequilibrium with advantageous haplotypes and thus invades the population, while a modifier allele that increases transformation rates generates linkage disequilibrium with disadvantageous haplotypes and thus goes extinct.  This is  the essence of the reduction principle.

Given this generality of the reduction principle for linear variation in transmission, we must ask why natural genetic systems in fact depart from the reduction principle and do not evolve to perfect transmission.  The  reduction principle guides us to the list of assumptions that must be violated in order for departures from reduction to occur.

Removal of the near-equilibrium assumption of the model allows several classes of departure from reduction.  Systems may be kept far from equilibrium by fluctuating selection, finite population size that produces drift, and {\it de novo} mutations.  Recombination has been found to evolve under a special form of cyclic selection \citep{Charlesworth:1976,Hamilton:1980}, and under newly changed directional selection \citep{Maynard:Smith:1988,
Charlesworth:1993,Barton:1995}, and in finite populations \citep{Felsenstein:1974,Felsenstein:and:Yokoyama:1976}.  Removal of the assumption of Mendelian segregation of the modifier locus, e.g. meiotic drive, allows the evolution of recombination \citep{Feldman:and:Otto:1991}.

Removal of the assumption of linear variation has produced many examples of departure from reduction.  Mechanistically, linear variation requires that the modifier controls the `hit' rate of a transforming process acting on the selected haplotypes, and that this process be the only transforming process, otherwise transmission probabilities between different genotypes will not be scaled equally.  Furthermore, for $\alpha$ to enter linearly in the recursion, the transformation process cannot allow multiple hits.

When spelled out thus, we see that the mechanistic requirements for linear variation are not realized  biologically.  Multiple hits are the norm in chiasma formation and in point mutation.  Multiple transformation processes are also the norm, as recombination, mutation, gene conversion, deletions, duplications, transpositions, etc. all happen in gamete formation.  

The case of multiple hits has been explored in a model of recombination modification for a multi-locus system \citep{Zhivotovsky:Feldman:and:Christiansen:1994}.  There, a refined reduction principle again holds.

The case of multiple simultaneous transformation processes acting during transmission has received a good deal of theoretical attention, and it is here that numerous departures from the reduction principle have been found.  A detailed review can be found in \citet{Feldman:Otto:and:Christiansen:1997}.  The cases exemplify the ``principal of partial control'' \citep[p. 149]{Altenberg:1984}: when a modifier gene has only partial control over the transformations occurring at selected loci, then it may be possible for this part of the transformations to evolve an increase.

The problem of characterizing exactly the conditions on variation in transmission that will evolve under modifier gene control remains an open question.  Its solution will require a deeper understanding of the differential properties of the spectra of non-negative matrices.

\section{Acknowledgments}
I thank my doctoral advisor Marc Feldman for introducing me to modifier theory, and Jean Doble for organizing ``Feldmania'', which inspired me to revisit the subject.  I thank Shmuel Friedland for introducing me to his work and the literature on the convexity of the spectral radius, which provided results crucial to the proof.


\end{document}